\documentclass[10pt,conference]{IEEEtran}
\usepackage{amsmath,amssymb}
\usepackage{graphicx}
\usepackage{mathrsfs}
\usepackage{tikz}
\usepackage{booktabs}
\usepackage{floatflt}
\usepackage{enumerate}
\usepackage{psfrag}	
\usepackage{array}
\usepackage{multirow,hhline}
\usepackage{exscale}
\usepackage{color}
\usepackage{colortbl}
\usepackage{epsfig}
\usepackage{cite}
\usepackage{amsthm}

\newcommand{\mat}[1]{\ensuremath{\mathbf{#1}}}

\newtheorem{theorem}{Theorem}
\newtheorem{lemma}{Lemma}
\newtheorem{definition}{Definition}

\newtheorem{corollary}{Corollary}

\begin{document}



\IEEEoverridecommandlockouts
\title{Achievable Rates and Upper bounds for the Interference Relay Channel}
\author{
\IEEEauthorblockN{Anas Chaaban and Aydin Sezgin}
\thanks{%
The authors are with the Emmy-Noether Research Group on Wireless Networks, 
Institute of Telecommunications and Applied Information Theory, 
Ulm University, 89081, Ulm, Germany,
Email: {anas.chaaban@uni-ulm.de, aydin.sezgin@uni-ulm.de}.
This work is supported by the German Research Foundation, Deutsche
Forschungsgemeinschaft (DFG), Germany, under grant SE 1697/3.%
}
}

\maketitle


\begin{abstract}
The two user Gaussian interference channel with a full-duplex relay is studied. By using genie aided approaches, two new upper bounds on the achievable sum-rate in this setup are derived. These upper bounds are shown to be tighter than previously known bounds under some conditions. Moreover, a transmit strategy for this setup is proposed. This strategy utilizes the following elements: Block Markov encoding combined with a Han-Kobayashi scheme at the sources, decode and forward at the relay, and Willems' backward decoding at the receivers. This scheme is shown to achieve within a finite gap our upper bounds in certain cases. 
\end{abstract}

\section{Introduction}

Relaying is an important strategy used to improve the performance in wireless networks. It can be used to overcome coverage problems, and furthermore, the use of relays can increase the achievable rate in a network. This fact can be seen in \cite{CoverElgamal} where the capacity of the relay channel consisting of a source, a relay, and a destination was studied, and it was shown that higher rates are achievable compared to the classical point to point channel.

By including one more transmit-receive pair to the point to point channel, we face an inevitable phenomenon in wireless networks, that is interference. This setup is known as the interference channel (IC) and has been the topic of intensive study for decades \cite{Carleial}.
Relaying can also be utilized as a means of cooperation in the IC, and the obtained setup is known as the interference channel with relay (IC-R). This setup has been studied in different variants: e.g. the IC with a full-duplex causal relay \cite{MaricDaboraGoldsmith,TianYener,SahinErkip}, and the IC with a cognitive relay \cite{SahinErkip_Cognitive,SridharanVishwanathJafarShamai}. In both variants, the impact of relaying on the system performance was analyzed, by studying upper bounds and achievable rate regions. However, same as for the IC, the capacity of the IC-R remains an open problem.

Several recent works study special cases of the IC with a full-duplex relay (IC-FDR), e.g. strong/weak source-relay links and strong interference. For instance, in \cite{TianYener} new upper bounds were developed for the IC with a potent relay, i.e. a relay that has no power constraint. Clearly, an IC with a potent relay provides an upper bound for the IC-FDR with a power constraint at the relay. The upper bounds given cover the case of weak interfering and source-relay links, and the case of strong interfering links. In \cite{MaricDaboraGoldsmith_GeneralizedRelaying}, an achievable scheme for the IC-FDR that uses block Markov encoding at the sources and decode and forward at the relay was proposed. In \cite{SahinErkip}, an achievable scheme similar to that in \cite{MaricDaboraGoldsmith_GeneralizedRelaying} was studied, with an additional component, that is rate splitting at the sources. The performance of this scheme is analyzed for the case when the source-relay links are strong, and thus, decode and forward at the relay does not limit the achievable rates. The IC-FDR with strong interference was studied in \cite{MaricDaboraGoldsmith}. A new upper bound was given, and this new bound was compared to an achievable rate in an IC-FDR with strong interference.

In this paper, we study the IC-FDR and establish two new upper bounds on the achievable sum-rate in this setup, based on genie aided approaches. One of our bounds is tighter than the cut-set bounds and the upper bound in \cite{MaricDaboraGoldsmith} at moderate to high power. Moreover, compared  to these bounds that require optimization, our bound is computable in closed form. The second upper bound we provide is relevant for the IC-FDR with weak interference.

We also provide an achievable scheme, that is a simplified version of the scheme in \cite{SahinErkip}. This scheme combines super-position block Markov encoding and Rate Splitting at the sources, decode and forward at the relay, and Willems' backward decoding at the destinations \cite{Willems}. If the IC-FDR has strong source-relay links, we show that this scheme achieves rates within a finite gap to the given upper bounds. In this case, the rate gain obtained by using the relay can be clearly seen from the expressions of the achievable rates. Moreover, we show that regardless of the strength of the source-relay links, when the relay-destination links are weak then using a Han-Kobayashi scheme \cite{HanKobayashi} as described in \cite{EtkinTseWang} already achieves rates within a constant gap to the developed upper bounds, 

Throughout the paper, we will use the following notations. We use $x^n$ to denote the sequence $(x_1,\dots,x_n)$. We denote $\frac{1}{2}\log(1+x)$ by $C(x)$. For $\alpha\in[0,1]$, $\bar{\alpha}=1-\alpha$.

\section{Model}
We consider a Gaussian IC-FDR as shown in Figure \ref{2UserICFDR}. Transmitter $i$ needs to communicate a message $m_i$ uniformly distributed over $\mathcal{M}_i\triangleq\{1,\dots,2^{nR_i}\}$ to its respective receiver. Each transmitter encodes its message to an $n$-symbol codeword $X_i^n$, and transmits this codeword. At time instant $k$, the input output equations of this setup are given by
\begin{eqnarray*}
y_1(k)&=&h_{11}x_1(k)+h_{21}x_2(k)+h_{r1}x_r(k)+z_1(k)\nonumber\\
y_2(k)&=&h_{22}x_2(k)+h_{12}x_1(k)+h_{r2}x_r(k)+z_2(k)
\end{eqnarray*}
\begin{eqnarray*}
y_r(k)&=&h_{1r}x_1(k)+h_{2r}x_2(k)+z_r(k).
\end{eqnarray*}
The coefficient $h_{ij}\geq0$ represents the channel gain from transmitter $i$ to receiver $j$, $i,j\in\{1,2\}$. The channels to and from the relay are denoted by $h_{ir}$ and $h_{ri}$ respectively. $x_r(k)$ is the transmit signal at the relay at time instant $k$. The relay is causal, which means that $x_r(k)$ is only a function of the previous observations of $X_1$ and $X_2$ at the relay, i.e. 
\begin{equation}
\label{Causality}
x_r(k)=f_r(y_r^{k-1}).
\end{equation}
The source and relay signals must satisfy $\mathbb{E}[X_j^2]\leq P$, $j\in\{1,2,r\}$. The noise at the receivers and the relay is assumed to be of zero-mean and unit-variance $z_1,z_2,z_r\in\mathcal{N}(0,1)$.

\begin{figure}[ht]
\centering{
\includegraphics[width=0.7\columnwidth]{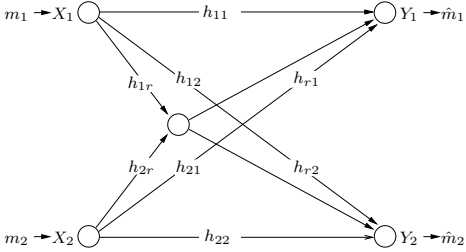}
}
\caption{The 2 user IC-FDR}
\label{2UserICFDR}
\end{figure}

Receiver $i$ decodes $\hat{m}_i=g_i(Y_i^n)$. The whole procedure defines a code denoted $(2^{nR_1},2^{nR_2},n)$. An error occurs if $\hat{m}_i\neq m_i$, and the average probability of error $P_e=P\left(\hat{m}_i\neq m_i: i\in\{1,2\}\right)$. 

A rate pair $(R_1,R_2)$ is said to be achievable if there exists a sequence of $(2^{nR_1},2^{nR_2},n)$ codes such that $P_e\to0$ as $n\to\infty$, and the capacity region $\mathcal{C}$ of the IC-FDR is the closure of the set of these achievable rate pairs.

\section{Known Upper Bounds}

The cut-set bound \cite{CoverThomas} is given by the following lemma.
\begin{lemma}
\label{Lemma:CutSet}
$\mathcal{C}\subset\mathcal{R}_{CS}\triangleq\bigcup_{\mat{A}\succeq0}\mathcal{R}_{cs}(\mat{A}),$ where $\mathcal{R}_{cs}(\mat{A})$ denotes the set of rate pairs $(R_1,R_2)$ that satisfy
\begin{align}
\label{Cut-Set-1}
R_1&\leq\min\{I(X_1,X_r;Y_1|X_2),I(X_1;Y_1,Y_r|X_2,X_r)\}\\
\label{Cut-Set-2}
R_2&\leq\min\{I(X_2,X_r;Y_2|X_1),I(X_2;Y_2,Y_r|X_1,X_r)\}\\
\label{Cut-Set-3}
R_1&+R_2\leq\min\{I(X_1,X_2,X_r;Y_1,Y_2),\nonumber\\
&\hspace{3.3cm}I(X_1,X_2;Y_1,Y_2,Y_r|X_r)\},
\end{align}
where $(X_1,X_2,X_r)$ are jointly Gaussian with covariance matrix
\begin{equation}
\label{CovMat}
\mat{A}=\left(\begin{array}{ccc}
P_1 	& 0 	& \rho_1\sqrt{P_1P_r}\\
0   	& P_2	& \rho_2\sqrt{P_2P_r}\\
\rho_1\sqrt{P_1P_r} & \rho_2\sqrt{P_2P_r} & P_r\end{array}\right),
\end{equation}
$\rho_1,\rho_2\in[-1,1]$ and $P_j\leq P\ \forall j\in\{1,2,r\}$.
\end{lemma}

In the following, we use $\mat{A}$ to denote a covariance matrix that satisfies the conditions in (\ref{CovMat}) without explicitly mentioning them. According to the cut-set bound, the maximum achievable sum-rate is bounded as follows:
\begin{corollary}
$R_1+R_2\leq\bar{R}_{CS}\triangleq\max_{(R_1^*,R_2^*)\in\mathcal{R}_{CS}}R_1^*+R_2^*$\end{corollary}

The first term in the sum rate cut-set bound (\ref{Cut-Set-3}) was tightened in \cite{MaricDaboraGoldsmith}. This sum rate upper bound is given in the following lemma.
\begin{lemma}[\cite{MaricDaboraGoldsmith}]
\label{Maric}\begin{eqnarray*}
R_1+R_2\leq R_m(\mat{A})\triangleq\min_{\{d_i\}_{i=1}^5}I(X_1,X_2,X_r;Y_1,Y_{1g})
\end{eqnarray*}
where $Y_{1g}=d_1X_1+d_2X_2+d_5X_r+d_3Z_1+d_4\tilde{Z}_1$, $(X_1,X_2,X_r)$ are jointly Gaussian with covariance matrix $\mat{A}$, $\tilde{Z}_1\sim\mathcal{N}(0,1)$ independent of all other variables, and $d_i$, $i\in\{1,\dots,5\}$ satisfy
\begin{align*}
&(1/h_{21}+v(d_3-d_2/h_{21}))^2+(v d_4)^2\leq1,\\
&d_5=(h_{r2}-u h_{r1})/v,\quad u=(1-v d_2)/h_{21},
\end{align*}
for some $u,v\in\mathbb{R}$, $v\neq0$.
\end{lemma}
Define the region $\mathcal{R}_m(\mat{A})\triangleq\{(R_1,R_2):R_1+R_2\leq R_m(\mat{A})\}$. The following corollaries are immediate conclusions from Lemma \ref{Maric}.
\begin{corollary}
$R_1+R_2\leq \bar{R}_M\triangleq\max_{\mat{A}\succeq0}R_m(\mat{A})$.
\end{corollary}
\begin{corollary}
$\mathcal{C}\subset\mathcal{R}_M\triangleq\bigcup_{\mat{A}\succeq0}\mathcal{R}_{cs}(\mat{A})\cap\mathcal{R}_m(\mat{A})$.
\end{corollary}
 
For further upper bounds, one might refer to \cite{TianYener} where two new upper bounds were introduced by using a potent relay approach (a relay with no power constraint), which clearly serves as an upper bound for the capacity of the IC-FDR.

It can be easily seen that the upper bound in Lemma \ref{Maric} can be written as $R_1+R_2\leq \log(P)+o(\log(P))$. Similar argument holds for the cut-set bounds. In the following section, we give a sum rate upper bound that is tighter than both at high $P$.

\section{New upper bounds}

The first upper bound is motivated by results in \cite{CadambeJafar_ImpactOfRelays} that show that (causal) relays can not increase the degrees of freedom of a (fully connected) wireless network. The upper bound we provide next agrees with this result as it can be written as $R_1+R_2\leq\frac{1}{2}\log(P)+o(\log(P))$.
\begin{theorem}
\label{SumBound1}
$R_1+R_2\leq \bar{R}_{s1}(\rho_1,\rho_2)$ where $\rho_1^2+\rho_2^2\leq1$ and
\begin{align*}
\bar{R}_{s1}(\rho_1,\rho_2)&\triangleq C\left(\frac{h_{22}^2P}{1+\max\{h_{21}^2,h_{2r}^2\}P}\right)+C\left(\frac{h_{2r}^2}{h_{21}^2}\right)\\ 
&\hspace{-1cm}+C\left(P\left(h_{11}^2+h_{21}^2+h_{r1}^2+2h_{r1}h_{11}\rho_1+2h_{r1}h_{21}\rho_2\right)\right)
\end{align*}
\end{theorem}
\begin{proof}
(Sketch) A genie gives $(m_1,h_{2r}X_2^n+Z_r^n,\tilde{Z}^n)$ to the second receiver where $\tilde{Z}^n=Z_1^n-\frac{h_{21}}{h_{2r}}Z_r^n$. We can show that $$h(h_{2r}X_2^n+Z_r^n|m_1,\tilde{Z}^n)\leq h(Y_1^n|m_1)-n\log(h_{21}/h_{2r}).$$
Moreover, it holds that $ h(Y_2^n|m_1,h_{2r}X_2^n+Z_r^n,\tilde{Z}^n)\leq$ $$\frac{n}{2}\log\left(2\pi e\left(1+\frac{h_{22}^2P}{1+\max\{h_{21}^2,h_{2r}^2\}P}\right)\right),$$
and $h(Z_r^n|\tilde{Z}^n)=\frac{n}{2}\log\left(2\pi e (h_{2r}^2)/(h_{2r}^2+h_{21}^2)\right).$
Now, using $n(R_1+R_2-2\epsilon_n)\leq I(m_1;Y_1^n)+I(m_2;Y_2^n,m_1,h_{2r}X_2^n+Z_r^n,\tilde{Z}^n)$ the result follows.
\end{proof}
\begin{corollary}
$R_1+R_2\leq \bar{R}_{S1}\triangleq\max_{\rho_1^2+\rho_2^2\leq1}\bar{R}_{s1}(\rho_1,\rho_2)=$
\begin{align*}
&C\left(P\left(h_{11}^2+h_{21}^2+h_{r1}^2+2h_{r1}\sqrt{h_{11}^2+h_{21}^2}\right)\right)\\
&+C\left(\frac{h_{22}^2P}{1+\max\{h_{21}^2,h_{2r}^2\}P}\right)+C\left(\frac{h_{2r}^2}{h_{21}^2}\right)
\end{align*}
\end{corollary}
Notice that $\bar{R}_{S1}$ is computable in closed form, compared to $\bar{R}_{M}$ which requires minimization over the variables $d_1,\dots,d_5$ and maximization over $\mat{A}$.
Define the region $\mathcal{R}_{s1}(\mat{A})\triangleq\{(R_1,R_2):R_1+R_2\leq \bar{R}_{s1}(\rho_1,\rho_2),\ \rho_1,\rho_2 \text{ as in }\mat{A}\}$, then we have:
\begin{corollary}
$\mathcal{C}\subset\mathcal{R}_{S1}\triangleq\bigcup_{\mat{A}\succeq0}\mathcal{R}_{cs}(\mat{A})\cap\mathcal{R}_{s1}(\mat{A})$.
\end{corollary}

Now we provide another bound that is inspired from the weak interference upper bound of the IC in \cite{EtkinTseWang}.

\begin{theorem}
\label{SumBound2}
$R_1+R_2\leq\bar{R}_{s2}(\mat{A})$ where
\begin{align*}
\bar{R}_{s2}(\mat{A})\triangleq C\left(\Theta_{12}\right)+C(\Theta_{21})+C\left(\frac{h_{1r}^2}{h_{12}^2}\right)+C\left(\frac{h_{2r}^2}{h_{21}^2}\right),
\end{align*}
$\sigma_{i}^2=h_{ij}^2/(h_{ij}^2+h_{ir}^2)$, and
\begin{align*}
\Theta_{ji}&=h_{ji}^2P_j+h_{ri}^2(1-\rho_i^2)P_r+2h_{ji}h_{ri}\rho_j\sqrt{P_jP_r}\\
&\ \ +\frac{\sigma_{i}^2(h_{ii}\sqrt{P_i}+h_{ri}\rho_i\sqrt{P_r})^2}{\sigma_{i}^2+h_{ij}^2P_i}
\end{align*}
\end{theorem}
\begin{proof}
(Sketch) We give the genie information $(S_j^n,X_{r1})$ to receiver $j$ where $S_j^n=h_{jk}X_j^n+W_j^n$ for $j,k\in\{1,2\}$ and $j\neq k$, $W_j\sim\mathcal{N}\left(0,\sigma_{j}^2\right)$, $\sigma_{j}^2=h_{jk}^2/(h_{jk}^2+h_{jr}^2)$, and $X_{r1}$ is the first symbol of the relay transmit sequence $X_r^n$. Then, we show that $h(Y_2^n|m_2,X_{r1})\geq h(h_{12}X_1^n+Z_2^n|\tilde{Z}_2^n,X_{r1})=h(h_{12}X_1^n+V_2^n|X_{r1})=h(S_1^n|X_{r1})$. This follows by adding one condition $\tilde{Z}_2^n=Z_r^n-h_{1r}Z_2^n/h_{12}$ which reduces entropy, and then arguing that knowing $m_2,\tilde{Z}_2^n,X_{r1},Y_2^{i-1}$ we can construct all $(X_{r2},\dots,X_{ri})$ and that $X_{r1}$ is independent of $m_2$ due to causality. Using similar arguments as in \cite[Lemma 6]{AnnapureddyVeeravalli}, we have $V_2\sim\mathcal{N}(0,\sigma_1^2)$.
It follows that $n(R_1+R_2-2\epsilon_n)\leq h(Y_1^n|S_1^n)+h(Y_2^n|S_2^n)-h(W_1^n)-h(W_2^n)$ and the result follows.
\end{proof}
Thus, the following corollary follows.
\begin{corollary}
$R_1+R_2\leq \bar{R}_{S2}\triangleq\max_{\mat{A}\succeq0}\bar{R}_{s2}(\mat{A})$
\end{corollary}
In Figures \ref{Plot1} and \ref{Plot2}, we plot the sum rate upper bounds $\bar{R}_{CS}$, $\bar{R}_M$, $\bar{R}_{S1}$, and $\bar{R}_{S2}$ for comparison. Figure \ref{Plot1} shows the case where the interfering links are stronger than the direct links. In this case, it can be seen that $\bar{R}_{S1}$ is lower than all other bounds at moderate to high $P$. We also observe that in this case $\bar{R}_{S2}$ is not relevant.
\begin{figure}[ht]
\centering
\psfragscanon
\psfrag{x}[t]{P(dB)}
\psfrag{y}[b]{$R_1+R_2$}
\psfrag{CSB}{\scriptsize{$\bar{R}_{CS}$}}
\psfrag{IMB}{\scriptsize{$\bar{R}_M$}}
\psfrag{T1B}{\scriptsize{$\bar{R}_{S1}$}}
\psfrag{T2B}{\scriptsize{$\bar{R}_{S2}$}}
\includegraphics[width=\columnwidth,height=5cm]{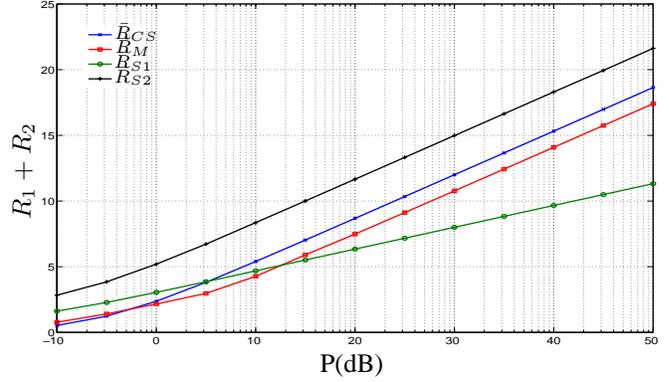}
\caption{Sum rate upper bounds for an IC-FDR with: $h_{11}=h_{22}=h_{r2}=1$, $h_{r1}=2$, $h_{12}=2$, $h_{21}^2=5$, $h_{1r}^2=h_{2r}^2=10$.}
\label{Plot1}
\end{figure}
Figure \ref{Plot2} shows the case where the interfering, source-relay, and relay-destination links are weak compared to the direct links. In this case, it can be seen that $\bar{R}_{S2}$ becomes relevant, since it is lower than $\bar{R}_{S1}$ at low $P$. It is slightly higher than $\bar{R}_M$ at low $P$. However, we consider this bound as it has the advantage that it involves less optimization steps. Moreover, as will be seen later, it is useful while calculating the gap to the achievable rate.
\begin{figure}[ht]
\centering
\psfragscanon
\psfrag{x}[t]{P(dB)}
\psfrag{y}[b]{$R_1+R_2$}
\psfrag{CSB}{\scriptsize{$\bar{R}_{CS}$}}
\psfrag{IMB}{\scriptsize{$\bar{R}_M$}}
\psfrag{T1B}{\scriptsize{$\bar{R}_{S1}$}}
\psfrag{T2B}{\scriptsize{$\bar{R}_{S2}$}}
\includegraphics[width=\columnwidth,height=5cm]{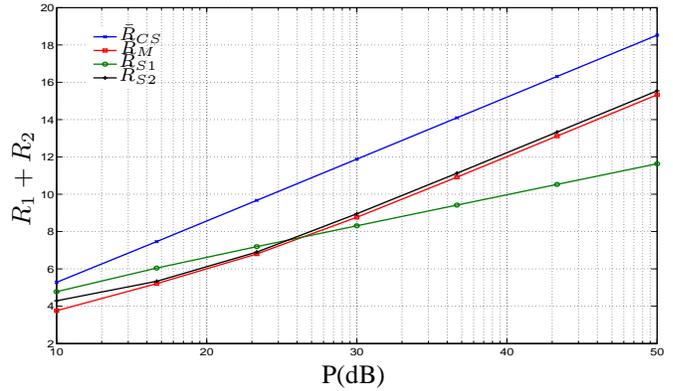}
\caption{Sum rate upper bounds for an IC-FDR with: $h_{11}=h_{22}=2$, $h_{r1}=h_{r2}=0.2$, $h_{12}=h_{21}=0.5$, $h_{1r}=h_{2r}=0.2$.}
\label{Plot2}
\end{figure}

\section{An Achievable Rate Region}
An achievable scheme for the IC-FDR was proposed in \cite{SahinErkip} that combines block Markov encoding, rate splitting, and backward decoding schemes. In this section, we give an achievable scheme similar to that in \cite{SahinErkip} with some simplification. We will provide a sketch of this achievable scheme.

The sources use super-position block Markov encoding, i.e. in a window of $B$ blocks, each source sends $B-1$ messages. The signal $x_i^n(b)$ sent by user $i$ in each block $b\in\{1,\dots,B\}$ is a super-position of codewords $u_i^n(b)$ and $u_i^n(b-1)$ from blocks $b$ and $b-1$ respectively and has power $P$. Moreover, each $u_i^n(b)$ is a  super-position of two codewords, $u_{i,c}^n(b)$ and $u_{i,p}^n(b)$ carrying common and private messages respectively.

In block $b$, The relay decodes $u_{1,c}^n(b)$, $u_{2,c}^n(b)$, $u_{1,p}^n(b)$, and $u_{2,cp}^n(b)$, and re-transmits them in the next block using power allocation parameters $\mu,\nu,\eta\in[0,1]$. This results in the following rate constraint at the relay
\begin{align}
\label{RelayRateConstraints1}
&a_1R_{1p}+a_2R_{1c}+a_3R_{2p}+a_4R_{2c}\leq\nonumber\\
&\hspace{0.8cm}C\left(P\left(\bar{\alpha}(a_1\bar{\gamma}+a_2\gamma)h_{1r}^2 +\bar{\beta}(a_3\bar{\delta}+a_4\delta)h_{2r}^2\right)\right)
\end{align}
for all $(a_1,a_2,a_3,a_4)\in\{0,1\}^4$, for some $\alpha,\beta,\gamma,\delta\in[0,1]$ denoting power allocation parameters at the sources. $R_{i,p}$ and $R_{i,c}$ are the rates of the private and common messages of source $i$ respectively. Thus the total rate achieved by each source is $R_i=R_{i,c}+R_{i,p}$.

The receivers use Willems' backward decoding to decode the messages starting from block $B$. In each block $b$, each receiver subtracts the interference that was already decoded in block $b+1$ and then proceeds to decode its private message and both common messages treating the remaining interference as noise as in \cite{EtkinTseWang}. Thus, the achievable common message rates lie in the intersection of the two regions $\mathcal{B}_1$ and $\mathcal{B}_2$ given by
\begin{align}
\label{CommonRateConstraints1}
\hspace{-0.3cm}\mathcal{B}_1=\left\{
\begin{array}{ll}
(R_{c1},R_{c2}):&\\
R_{c1}&\hspace{-0.3cm}\leq
C\left(\frac{\tilde{h}_{11c}^2 P}{1+\tilde{h}_{11p}^2P+\tilde{h}_{21p}^2P+h_{21}^2\bar{\beta}\bar{\delta}P}\right)\\
R_{c2}&\hspace{-0.3cm}\leq C\left(\frac{\tilde{h}_{21c}^2P}{1+\tilde{h}_{11p}^2P+\tilde{h}_{21p}^2P+h_{21}^2\bar{\beta}\bar{\delta}P}\right)\\
R_{c1}+R_{c2}&\hspace{-0.3cm}\leq C\left(\frac{\tilde{h}_{11c}^2P+\tilde{h}_{21c}^2P}{1+\tilde{h}_{11p}^2P+\tilde{h}_{21p}^2P+h_{21}^2\bar{\beta}\bar{\delta}P}\right)\end{array}\right\}
\end{align}
\begin{align}
\label{CommonRateConstraints2}
\hspace{-0.3cm}\mathcal{B}_2=\left\{
\begin{array}{ll}
(R_{c1},R_{c2}):&\\
R_{c1}&\hspace{-0.3cm}\leq C\left(\frac{\tilde{h}_{12c}^2P}{1+\tilde{h}_{12p}^2P+\tilde{h}_{22p}^2P+h_{12}^2\bar{\alpha}\bar{\gamma}P}\right)\\
R_{c2}&\hspace{-0.3cm}\leq C\left(\frac{\tilde{h}_{22c}^2P}{1+\tilde{h}_{12p}^2P+\tilde{h}_{22p}^2P+h_{12}^2\bar{\alpha}\bar{\gamma}P}\right)\\
R_{c1}+R_{c2}&\hspace{-0.3cm}\leq C\left(\frac{\tilde{h}_{12c}^2P+\tilde{h}_{22c}^2P}{1+\tilde{h}_{12p}^2P+\tilde{h}_{22p}^2P+h_{12}^2\bar{\alpha}\bar{\gamma}P}\right)\end{array}\right\}
\end{align}
where we use
\begin{align*}
\tilde{h}_{1jp}=(h_{1j}\sqrt{\alpha\bar{\gamma}}+h_{rj}\sqrt{\eta\bar{\mu}}),\ \  \tilde{h}_{1jc}=(h_{1j}\sqrt{\alpha\gamma}+h_{rj}\sqrt{\eta\mu}),\\ \tilde{h}_{2jp}=(h_{2j}\sqrt{\beta\bar{\delta}}+h_{rj}\sqrt{\bar{\eta}\bar{\nu}}),\ \  \tilde{h}_{2jc}=(h_{2j}\sqrt{\beta\delta}+h_{rj}\sqrt{\bar{\eta}\nu}),
\end{align*}
for receiver $j\in\{1,2\}$. Moreover, the following private message rate constraints must be satisfied,
\begin{eqnarray}
\label{PrivateRateConstraints1}
R_{p1}&\leq&C\left(\frac{\tilde{h}_{11p}^2P}{1+\tilde{h}_{21p}^2P+h_{21}^2\bar{\beta}\bar{\delta}P}\right)\\
\label{PrivateRateConstraints2}
R_{p2}&\leq&C\left(\frac{\tilde{h}_{22p}^2P}{1+\tilde{h}_{12p}^2P+h_{12}^2\bar{\alpha}\bar{\gamma}P}\right).
\end{eqnarray}

\begin{definition}
Denote by $\tilde{\mathcal{R}}(\zeta)$ the following set 
\begin{align*}
\tilde{\mathcal{R}}(\zeta)\triangleq\left\{(R_1,R_2): R_1=R_{p1}+R_{c1}, R_2=R_{p2}+R_{c2}\right\}
\end{align*}
with $R_{p1}$, $R_{c1}$, $R_{p2}$, $R_{c2}$ satisfying (\ref{RelayRateConstraints1}), (\ref{PrivateRateConstraints1}), and (\ref{PrivateRateConstraints2}) and $(R_{c1},R_{c2})\in\mathcal{B}_1\cap\mathcal{B}_2$, for a given power allocation vector $\zeta=(\alpha,\beta,\gamma,\delta,\eta,\mu,\nu)\in[0,1]^7$.
\end{definition}

Now we can state the following theorem.
\begin{theorem}
\label{AchievableRegion}
$\mathcal{C}\supset\mathcal{R}\triangleq\bigcup_{\zeta\in[0,1]^7}\tilde{\mathcal{R}}(\zeta).$
\end{theorem}
Figure \ref{Reg1} shows the achievable rate region as given in Theorem \ref{AchievableRegion} with the outer bounds for comparison. As shown, our sum-rate outer bound $\mathcal{R}_{S1}$ is tighter than the other outer bounds in this case.

\begin{figure}[ht]
\centering
\psfragscanon
\psfrag{x}[t]{$R_1$}
\psfrag{y}{$R_2$}
\psfrag{CSB}{\scriptsize{$\mathcal{R}_{CS}$}}
\psfrag{IMB}{\scriptsize{$\mathcal{R}_M$}}
\psfrag{SIB}{\scriptsize{$\mathcal{R}_{S1}$}}
\psfrag{ACH}{\scriptsize{$\mathcal{R}$}}
\includegraphics[width=0.8\columnwidth,height=5cm]{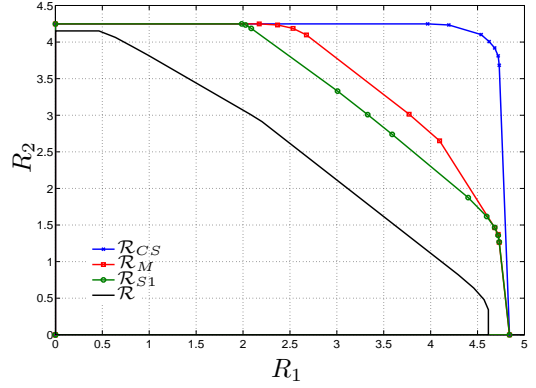}
\caption{Rate region outer and inner bounds for an IC-FDR with: $h_{11}=h_{22}=1$, $h_{r1}=2$, $h_{r2}=1$, $h_{12}=2$, $h_{21}^2=5$, $h_{1r}^2=h_{2r}^2=10$, $P=20dB$.}
\label{Reg1}
\end{figure}

\section{The Symmetric IC-FDR}
The symmetric IC-FDR has $h_{11}=h_{22}=h_d$, $h_{21}=h_{12}=h_c$, and $h_{r1}=h_{r2}=h_r$. 
By fixing $\alpha=\beta=\eta=0.5$, $\mu=\nu$, and choosing $\bar{\gamma}=\bar{\delta}=\frac{2}{h_c^2P}$, we can show that the following symmetric rate $R_1=R_2=R_s$ is achievable.
\begin{corollary}
\label{SymRate}
$(R_s,R_s)\in\mathcal{C}$ where
\begin{align*}
R_{s}&=\min\left\{\frac{1}{2}C(h_{sr}^2P),\min\{R_A,R_B,R_C\}-\frac{1}{2}\log(3)\right\},
\end{align*}
\begin{align*}
R_A&=\frac{1}{2}C\left(2+\Phi\right)+\frac{1}{2}C\left(2+\Phi+\Psi+\Omega\right)\nonumber\\
R_B&=C\left(2+\Phi+\Psi\right)\\
R_C&=C\left(2+\Phi+\Omega\right),
\end{align*} with $\Phi=h_d^2/h_c^2$, $\Psi=\left(\sqrt{{h_d^2P}/{2}-\Phi}+\sqrt{{h_r^2P}/{2}}\right)^2$, and $\Omega=\left(\sqrt{{h_c^2P}/{2}-1}+\sqrt{{h_r^2P}/{2}}\right)^2$.
\end{corollary}

Notice that the achievable symmetric rate in an IC is also achievable in the IC-FDR, by simply ignoring the relay and using the IC scheme in \cite{EtkinTseWang}.
\begin{theorem}
\label{IC_Achievable}
$(R_s^{IC},R_s^{IC})\in\mathcal{C}$ where
\begin{align*}
R_{s}^{IC}&=\left\{\begin{array}{ll}
\min\left\{R^{IC}_A,R^{IC}_B\right\}&\text{if } h_c^2\leq h_d^2\\
\min\left\{R^{IC}_C,R^{IC}_D\right\}&\text{otherwise}
\end{array}\right.
\end{align*}
with 
\vspace{-0.4cm}
\begin{align*}
R^{IC}_A&=C\left(h_c^2P+{h_d^2P}/{h_c^2P}\right)-1/2,\\
2R^{IC}_B&=C(h_d^2P+h_c^2P)+C\left({h_d^2P}/{h_c^2P}\right)-1,\\
R^{IC}_C&=C(h_d^2P),\\
2R^{IC}_D&=C(h_d^2P+h_c^2P).
\end{align*}
\end{theorem}

\section{Gap Analysis}
In this section, we will bound the gap between the achievable symmetric rate and the upper bounds. The upper bound for the achievable symmetric rate is given by $$\bar{R}\triangleq\max_{(R,R)\in\mathcal{C}}R\leq\frac{1}{2}\min\left\{\bar{R}_M,\bar{R}_{S1},\bar{R}_{S2},\bar{R}_{CS}\right\}.$$
\subsection{Strong source-relay links}
We start by considering the achievable rate in corollary \ref{SymRate}. Assume that
\begin{align} 
\label{MABC}
\frac{1}{2}\log(h_{sr}^2P)\geq\min\{R_A,R_B,R_C\}-\frac{1}{2}\log(3).
\end{align}
Then, according to the term that dominates the minimization in (\ref{MABC}), the gap between $\bar{R}$ and $R_s$ satisfies:
\begin{align*}
&\bar{R}-R_{s}\leq\\
&\left\{\begin{array}{lcr}
\frac{3}{4}+\frac{1}{2}\log(3)+\frac{1}{2}C\left(\frac{h_{sr}^2}{h_{c}^2}\right)& \text{if }& R_B,R_C\geq R_A\\
 1+\frac{1}{2}\log(3)& \text{if }& R_A,R_C\geq R_B\\
1+\frac{1}{2}\log(3)+\frac{1}{2}\log(5)+C\left(\frac{h_{sr}^2}{h_c^2}\right)& \text{if }& R_A,R_B\geq R_C
\end{array}\right.
\end{align*}
Thus the upper and lower bounds are within a finite gap. However, this gap is not universal, it depends on $h_{sr}^2/h_{c}^2$.

\subsection{Weak relay-destination links}
In this case, it can be shown that by utilizing transmission schemes for the IC, i.e. ignoring the relay, we can achieve within a constant gap the upper bounds for any value of $h_{sr}$. To see this, assume that $h_r^2\leq\min\{h_d^2,h_c^2\}$, then it can be shown that the gap between $\bar{R}$ and $R_s^{IC}$ given in Theorem \ref{IC_Achievable} satisfies
\begin{align*}
\bar{R}-R_s\leq\left\{\begin{array}{lcr}
1.5+C(h_{sr}^2/h_c^2)&\text{if }&R_s=R^{IC}_A\\
7/8+0.5C(h_{sr}^2/h_c^2)&\text{if }&R_s=R^{IC}_B\\
1&\text{if }&R_s=R^{IC}_C\\
5/8+0.5C(h_{sr}^2/h_c^2)&\text{if }&R_s=R^{IC}_D
\end{array}\right.
\end{align*}

Consequently, if the IC-FDR has $h_r^2\leq\min\{h_d^2,h_c^2\}$, then by ignoring the relay and operating the IC-FDR as an IC, we achieve its sum capacity within a finite gap. Notice that in this case the value of $h_{sr}$ does not limit our achievable rates since we do not use the relay.

In Figure \ref{Gap_3d}, we plot the gap $\Delta=\bar{R}-\underline{R}$ as a function of $$a=\frac{\log(h_c^2P)}{\log(h_d^2P)}\quad \text{and}\quad b=\frac{\log(h_{sr}^2P)}{\log(h_d^2P)}$$
for an IC-FDR with $h_d=h_r=1,$ and $P=30dB$, where $$\underline{R}\triangleq\max\{R_s^{IC},\max_{(R,R)\in\mathcal{R}}R\}.$$ This plots shows the gap for channels with different $h_c$ and $h_{sr}$.
\begin{figure}[ht]
\centering
\psfragscanon
\psfrag{x}{$a$}
\psfrag{y}{$b$}
\includegraphics[width=0.8\columnwidth]{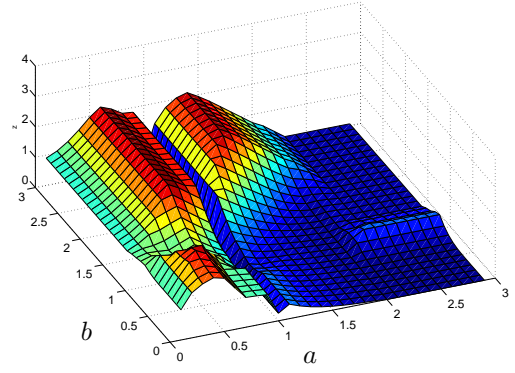}
\caption{The gap between the upper bounds and the lower bounds for the symmetric rate in an IC-FDR with $h_{d}=h_{r}=1$ and $P=30dB$ as a function of $a$ and $b$.}
\label{Gap_3d}
\end{figure}

\section{Conclusion}
\label{Section:Conclusion}
We have studied the interference channel with a full-duplex relay (IC-FDR). We derived two new upper bounds for this setup. These bounds improve previously known bounds for the IC-FDR. Furthermore, we studied the achievable rate in this setup. We derived an achievable rate region. Based on this rate region, the achievable symmetric rate in the symmetric IC-FDR is given. We showed that this achievable symmetric rate is within a finite gap to our upper bounds when the IC-FDR has strong source-relay links. Moreover, we showed that if the relay-destination links are weak, then the given upper bounds can be achieved within a constant gap by simply ignoring the relay.

\section{Acknowledgment}
The authors would like to show their appreciation to Deniz G\"und\"uz for fruitful discussions.

\bibliography{myBib}		

\end{document}